\documentclass[letter,runningheads,envcountsame,envcountsect]{llncs}
 
 \usepackage{amsmath}
\usepackage{todonotes}

\title{Algorithms and Bounds for Very Strong Rainbow Coloring}
\author{L. Sunil Chandran\inst{1} \and Anita Das\inst{2} \and Davis Issac\inst{3} \and Erik Jan van Leeuwen\inst{4}}
\institute{Department of Computer Science and Automation, Indian Institute of Science, Bangalore, India, \email{sunil@csa.iisc.ernet.in}. \and Infosys Ltd., \email{anita\_das01@infosys.com}. \and MPI f\"ur Informatik, Saarland Informatics Campus, Saarbr\"ucken, Germany, \email{dissac@mpi-inf.mpg.de}. \and Department of Information and Computing Sciences, Utrecht University, The Netherlands, \email{e.j.vanleeuwen@uu.nl}.}

\spnewtheorem{observation}[theorem]{Observation}{\bfseries}{\itshape}

\newcommand{\U}{\mathcal{U}}
\newcommand{\F}{\mathcal{F}}
\newcommand{\C}{\mathcal{C}}
\newcommand{\vsrc}{\mathbf{vsrc}}
\newcommand{\cvsrc}{VSRC}
\newcommand{\src}{\mathbf{src}}
\newcommand{\rc}{\mathbf{rc}}
\newcommand{\tw}{\mathbf{tw}}

\newcommand{\tvsrc}{{\sc 3-Vsrc}}
\newcommand{\twovsrc}{{\sc 2-Vsrc}}
\newcommand{\kvsrc}{{\sc $k$-Vsrc}}

\newcommand{\twocolor}{{\sc 2-Coloring}}
\newcommand{\cp}{\mathbf{cp}}
\newcommand{\sat}{{\sc Sat}}
\newcommand{\et}{E_{\mathsf{bridge}}}
\newcommand{\eeven}{E_{\mathsf{even}}}
\newcommand{\eodd}{E_{\mathsf{odd}}}
\newcommand{\eopp}{E_{\mathsf{opp}}}
\newcommand{\erem}{E_{\mathsf{rem}}}

\newcommand{\etal}{et al.}
\newcommand{\order}{\mathcal{O}}
\newcommand{\isn}{\mathbf{is}}
\newcommand{\ecc}{\mathbf{ecc}}
\newcommand{\vopp}{\mathsf{vopp}}
\newcommand{\eeopp}{\mathsf{eopp}}
\newcommand{\os}{\mathsf{OS}}
\newcommand{\g}{\mathsf{g}}
\newcommand{\subg}{\mathsf{S}}

\begin{document}

\maketitle

\begin{abstract}
A well-studied coloring problem is to assign colors to the edges of a graph $G$ so that, for every pair of vertices, all edges of at least one shortest path between them receive different colors. The minimum number of colors necessary in such a coloring is the strong rainbow connection number ($\src(G)$) of the graph. When proving upper bounds on $\src(G)$, it is natural to prove that a coloring exists where, for \emph{every} shortest path between every pair of vertices in the graph, all edges of the path receive different colors. Therefore, we introduce and formally define this more restricted edge coloring number, which we call \emph{very strong rainbow connection number} ($\vsrc(G)$). 

In this paper, we give upper bounds on $\vsrc(G)$ for several graph classes, some of which are tight. These immediately imply new upper bounds on $\src(G)$ for these classes, showing that the study of $\vsrc(G)$ enables meaningful progress on bounding $\src(G)$. Then we study the complexity of the problem to compute $\vsrc(G)$, particularly for graphs of bounded treewidth, and show this is an interesting problem in its own right. We prove that $\vsrc(G)$ can be computed in polynomial time on cactus graphs; in contrast, this question is still open for $\src(G)$. We also observe that deciding whether $\vsrc(G) = k$ is fixed-parameter tractable in $k$ and the treewidth of $G$. Finally, on general graphs, we prove that there is no polynomial-time algorithm to decide whether $\vsrc(G) \leq 3$ nor to approximate $\vsrc(G)$ within a factor $n^{1-\varepsilon}$, unless P$=$NP. 


\end{abstract}

\section{Introduction} \label{sec:intro}
The chromatic number is one of the most widely studied properties in graph theory. It has inspired a wealth of combinatorial and algorithmic results, as well as a host of variants. A variant that has recently attracted much interest is the \emph{rainbow connection number} of a graph, which is an edge coloring property introduced by Chartrand~\etal~\cite{chartrandrainbow} in 2008. Formally, the rainbow connection number $\rc(G)$ of a graph $G$ is the smallest number of colors needed such that there exists a coloring of $E(G)$ with these colors such that, for every pair of vertices, there exists at least one path $P$ between them, such that all edges of $P$ receive different colors. We also say that this path $P$ is \emph{rainbow colored}. The rainbow connection number has attracted much attention, and the exact number is known for a variety of simple graph classes~\cite{chartrandrainbow,chandran2012rainbow,sunil2012rainbow} and the complexity of computing this number was broadly investigated~\cite{ananth2011rainbow,basavaraju2014rainbow,chakraborty2011hardness,chandran2012rainbow,chandran2013inapproximability}. See also the surveys by Li~\etal~\cite{srcli,li2012rainbow,li2017}. Most recently, in ESA 2016, it was shown that for any $k \geq 2$, deciding whether $\rc(G)\le k$ (\textsc{$k$-Rc}) cannot be solved in $2^{o(n^{3/2})}$ or $2^{o(m/\log m)}$ time, where $n = |V(G)|$ and $m = |E(G)|$, unless ETH fails~\cite{Lauri-esa}.

To prove an upper bound on $\rc(G)$, the choice of the path $P$ that is rainbow colored is crucial. The analysis would seem simpler when we are able to choose $P$ as a shortest path between its two endpoints. This leads to the definition of the \emph{strong rainbow connection number} of a graph. Formally, the strong rainbow connection number $\src(G)$ of a graph $G$ is the smallest number of colors needed such that there exists a coloring of $E(G)$ with these colors such that, for every pair of vertices, there exists at least one \emph{shortest} path $P$ between them, such that all edges of $P$ receive different colors. Clearly, $\src(G) \geq \rc(G)$, and both parameters are at least the diameter of $G$. Moreover, $\rc(G) = 2$ if and only if $\src(G) = 2$~\cite{chakraborty2011hardness}. Nontrivial upper bounds on $\src(G)$ are known for several simple graph classes such as cycles, wheels, and complete bipartite graphs~\cite{chartrandrainbow} and block graphs~\cite{lauri2016chasing}. It is also known that deciding whether $\src(G) \leq k$ (\textsc{$k$-Src}) is NP-hard even for $k=2$~\cite{chakraborty2011hardness}. The problem of deciding whether $\src(G) \leq k$ remains NP-complete even for bipartite graphs and split graphs~\cite{ananth2011rainbow,keranen2014}. In fact, $\src(G)$ cannot be approximated in polynomial time within a factor $n^{1/2-\varepsilon}$ for any $\varepsilon > 0$, unless P$=$NP, even for split and bipartite graphs~\cite{ananth2011rainbow,keranen2014}.%
\footnote{\cite{ananth2011rainbow} and \cite{keranen2014} mention NP$\neq$ZPP as the complexity assumption but one can use P$\neq$NP because of \cite{zuckerman2006linear}.}

The lack of combinatorial bounds on $\src(G)$ for specific graph classes $G$ (the recent survey by Li and Sun~\cite{li2017} cites only three papers) is somewhat surprising compared to the vast literature for $\rc(G)$ (see the surveys~\cite{srcli,li2012rainbow,li2017}). Li and Sun~\cite{li2017} explain this by the fact that $\src(G)$ is not a monotone graph property, and thus investigating $\src(G)$ is much harder than investigating $\rc(G)$. Hence, it is a major open question to prove upper bounds on $\src(G)$.

In this paper, we make significant progress on this question. We observe that to prove upper bounds on $\src(G)$, it suffices to prove the existence of a coloring where all edges of not just one, but of \emph{all} shortest paths between two vertices receive different colors. Therefore, we define the \emph{very strong rainbow connection number} $\vsrc(G)$ of a graph $G$, which is the smallest number of colors for which there exists a coloring of $E(G)$ such that, for every pair of vertices and \emph{every} \emph{shortest} path $P$ between them, all edges of $P$ receive different colors. We call a coloring that achieves this property a \emph{very strong rainbow coloring} of the graph. We also call the problem of deciding whether $\vsrc(G) \leq k$ the {\kvsrc} problem. 

\paragraph{Our Results}
We prove the first combinatorial upper bounds on $\vsrc(G)$ for several graph classes. These immediately imply upper bounds on $\src(G)$ for the same graph classes. In particular, we show upper bounds that are linear in $|V(G)|$ (improving from the trivial bound of $|E(G)|$) if $G$ is a chordal graph, a circular arc graph, or a disk graph.
We also make progress on the following conjecture:
\begin{conjecture}[\cite{lauri2016chasing}]
For any connected graph $G$, $\src(G)\le |V(G)|-\chi(G)+1$ where $\chi(G)$ denotes the chromatic number of $G$.	
\end{conjecture}
We show that the conjecture holds for the class of chordal graphs in Lemma~\ref{cor:chordal}.

Conversely, we prove that a bound on $\vsrc(G)$ implies that $G$ should be highly structured: the neighborhood of every vertex can be partitioned into $\vsrc(G)$ cliques. For further details, we refer to Section~\ref{sec:combinatorial}.

In the second part of the paper, we address the computational complexity of {\kvsrc}. To start our investigation, we prove hardness results on general graphs.
\begin{theorem} \label{thm:hardness}
{\tvsrc} is NP-complete. Moreover, there is no polynomial-time algorithm that approximates $\vsrc(G)$ within a factor $|V(G)|^{1-\varepsilon}$ for any $\varepsilon >0$, unless P$=$NP.
\end{theorem}
This result implies that {\kvsrc} is not fixed-parameter tractable when parameterized by $k$, unless P$=$NP. In order to prove the theorem, we show a nontrivial connection to the clique partition number of a graph. 

We remark that, in contrast to the NP-complete \textsc{$2$-Rc} and \textsc{$2$-Src} problems, {\twovsrc} can be solved in polynomial time (see Section~\ref{sec:other_results} for the proof). Together with Theorem~\ref{thm:hardness}, this gives a dichotomy result for the complexity of {\kvsrc}.

\begin{proposition} \label{prp:two}
Let $G$ be any graph. Then {\twovsrc} can be decided in polynomial time.
\end{proposition}

We then study the complexity of determining $\vsrc(G)$ for graphs of bounded treewidth. This is a major open question also for $\src(G)$ and $\rc(G)$~\cite{lauri2016chasing}, which are only known to be solvable in polynomial time on graphs of treewidth~$1$. We mention that no results for graphs of higher treewidth are known, even for outerplanar or cactus graphs. However, for the slightly different problem of deciding whether an already given coloring forms a (strong) rainbow coloring of a given graph, a polynomial-time algorithm for cactus graphs and an NP-hardness result for outerplanar graphs are known~\cite{uchizawa2013rainbow}. With this in mind, we focus on cactus graphs and make the first progress towards understanding the complexity of rainbow coloring problems, in particular of computing $\vsrc(G)$, on graphs of treewidth~$2$ with the following result.

\begin{theorem} \label{thm:cactus}
Let $G$ be any cactus graph. Then $\vsrc(G)$ can be computed in polynomial time.
\end{theorem}
Our algorithm relies on an extensive characterization result for the behavior of very strong rainbow colorings on cactus graphs. Since a cactus graph consists of bridges, even cycles, and odd cycles, we analyze the behavior of any very strong rainbow coloring of the graph with respect to these structures. We show that color repetition can mostly occur only within an odd cycle or even cycle. 
Odd cycles can repeat some colors from outside but we characterize how they can be repeated.
However, our arguments are not sufficient to derive a completely combinatorial bound. Instead, we must find a maximum matching in a well-chosen auxiliary graph to compute the very strong rainbow connection number.

We also observe that $\vsrc(G)$ can be computed efficiently for graphs having bounded treewidth, when $\vsrc(G)$ itself is small. In contrast to known results for the (strong) rainbow connection number~\cite{eiben2016complexity}, we present an algorithm that does not rely on Courcelle's theorem. (See section~\ref{sec:other_results} for details.)

\begin{theorem} \label{fpt-kt}
{\kvsrc} is fixed-parameter tractable when parameterized by $k+t$, where $t-1$ is the treewidth of the input graph.
\end{theorem}

\paragraph{Preliminaries}
We consider simple, undirected graphs and use standard notation for graphs. Given a universe $\U=\left\{ x_1,x_2,\dots, x_n \right\}$ and a family $\F=\left\{ S_1,S_2,\dots,S_t \right\}$ of subsets of $\U$, the \emph{intersection graph} $G(\F)$ of $\F$ has vertex set $\{v_1,\ldots,v_t\}$, and there is an edge between two vertices $v_i,v_j$ if and only if $S_i \cap S_j \not= \emptyset$. We call $\F$ a \emph{representation} of $G(\F)$. An \emph{interval graph} is an intersection graph of intervals on the real line. The interval graph is \emph{proper} if it has a representation by intervals where no interval is properly contained in another. A \emph{circular arc graph} is an intersection graph of arcs of a circle. A chordal graph is an intersection graph of subtrees of a tree.
A \emph{block} of a graph is a maximal $2$-connected component. 
In a \emph{cactus graph}, each block of the graph is a cycle or an edge; equivalently, every edge belongs to at most one cycle.

For a graph $G$, let $\hat{G}$ denote the graph obtained by adding a new vertex $\hat{u}$ to $G$ such that $\hat{u}$ is adjacent to all vertices of $G$, i.e., $\hat{u}$ is a universal vertex in $\hat{G}$. 

Finally, we use $\omega(G)$ to denote the maximum size of any clique in graph $G$. We use $d(u,v)$ to denote the length of a shortest path between vertices $u$ and $v$.

\section{Combinatorial Results} \label{sec:combinatorial}
We show several upper and lower bounds on $\vsrc(G)$, both for general graphs and for graphs $G$ that belong to a specific graph class. Crucial in our analysis are connections between very strong rainbow colorings and decompositions of the input graph into cliques.
We use $\cp(G)$ to denote the \emph{clique partition number} (or clique cover number) of $G$, the smallest number of subsets of $V(G)$ that each induce a clique in $G$ and whose union is $V(G)$.
$\hat{G}$ used in the following lemma (defined in the preliminaries) is important for our hardness reductions.

\begin{lemma}\label{vsrc-l-cp2}
Let $G$ be any graph. Then 
\begin{enumerate}
\item $\src(G) \leq \vsrc(G) \leq \cp(G)(\cp(G)+1)/2$.
\item $\src(\hat{G}) \leq \vsrc(\hat{G}) \leq \cp(G)(\cp(G)+1)/2$. 
\end{enumerate}
\end{lemma}
\begin{proof}
Let $\C={C_1,\dots,C_r}$ be the set of cliques in an optimal clique partition of $G$; that is, $r = \cp(G)$. For a vertex $v$, let $c(v)$ denote the clique in $\C$ that contains $v$. We define the set of colors as $\mathcal{P}_{\leq 2}(\C) \setminus \{\emptyset\}$, the set of subsets of $\C$ of size~$1$ or~$2$. We then color any edge $uv \in E(G)$ by $\{c(u),c(v)\}$. For sake of contradiction, suppose that this does not constitute a very strong rainbow coloring of $G$. Then there exist two vertices $s,t \in V(G)$, a shortest path $P$ between $s$ and $t$, and two edges $uv, wx \in E(P)$ that received the same color. If $c(u) = c(v)$, then $c(w) = c(x)$, meaning that $P$ uses two edges of the same clique. Then $P$ can be shortcut, contradicting that $P$ is a shortest path between $s$ and $t$. Hence, $c(u) \not= c(v)$ and thus $c(w) \not= c(x)$. Without loss of generality, $c(u) = c(w)$ and thus $c(v) = c(x)$. Then either the edge $uw$ or the edge $vx$ will shortcut $P$, a contradiction. Hence, $\vsrc(G) \leq \cp(G)(\cp(G)+1)/2$ by the set of colors used.\\
To see the second part of the lemma, color edges $\hat{u}v$ incident on the universal vertex $\hat{u}$ in $\hat{G}$ by $c(v)$ in addition to the above coloring. Suppose this was not a very strong rainbow coloring of $\hat{G}$. Then there exists vertices $u,v$ such that $u\hat{u}v$ is a shortest path and $u\hat{u}$ and $v\hat{u}$ are colored the same. But then $u$ and $v$ are in the same clique $C_i$ in $\C$. But then $uv$ can shortcut $u\hat{u}v$, a contradiction.
\qed\end{proof}
%
The following lemma is more consequential for our upper bounds. We use $\isn(G)$ to denote the smallest size of the universe in any intersection graph representation of $G$, and $\ecc(G)$ to denote the smallest number of cliques needed to cover all edges of $G$. It is known that $\isn(G) = \ecc(G)$~\cite{roberts1985applications}.

\begin{lemma}
	\label{lemma:intersection}
Let $G$ be any graph. Then $\vsrc(G) \leq \isn(G) = \ecc(G)$.
\end{lemma}
\begin{proof}
Let $\;\U=\left\{ x_1,x_2,\dots, x_n \right\}$ be a universe and let $\F=\left\{ S_1,S_2,\dots,S_m \right\}$ be a family of subsets of $\U$, such that $G$ is the intersection graph of $\F$ and $|\U| = \isn(G)$. Let $v_i$ be the vertex of $G$ corresponding to the set $S_i$. We consider $x_1, x_2, \dots, x_n$ as colors, and color an edge between vertices $v_i$ and $v_j$ with any $x \in S_i \cap S_j$ (note that this intersection is nonempty by the presence of the edge). Suppose for sake of contradiction that this is not a very strong rainbow coloring of $G$. Then there exist two vertices $s,t \in V(G)$, a shortest path $P$ between $s$ and $t$, and two edges $v_iv_j$ and $v_av_b$ in $P$ that received the same color $x$. By the construction of the coloring, this implies that $x\in S_i\cap S_j\cap S_a\cap S_b$. Hence, $v_i,v_j,v_a,v_b$ induce a clique in $G$. But then the path $P$ can be shortcut, a contradiction.
\qed\end{proof}
A similar lemma for $\src(G)$ was proved independently by Lauri~\cite[Prop.~5.3]{lauri2016chasing}.

\begin{corollary}
Let $G$ be any graph. Then $\vsrc(G) \leq \min\{\lfloor |V(G)|^2/4 \rfloor, |E(G)|\}$.
\end{corollary}
\begin{proof}
Directly from $\ecc(G)\leq \min\{\lfloor |V(G)|^2/4 \rfloor, |E(G)|\}$ for any graph~\cite{erdos1966representation}.
\qed\end{proof}


\begin{corollary}
	\label{cor:chordal}
Let $G$ be any graph.
\begin{enumerate}
\item If $G$ is chordal, then $\src(G) \leq \vsrc(G) \leq |V(G)| - \omega(G) + 1$.
\item If $G$ is circular-arc, then $\src(G) \leq \vsrc(G) \leq |V(G)|$.
\item $\src(L(G)) \leq \vsrc(L(G)) \leq |V(G)|$, where $L(G)$ is the line graph of $G$.
\end{enumerate}
These bounds are (almost) tight in general.
\end{corollary}
\begin{proof}
	In each of the three cases, we express the graph as an intersection graph over a suitable universe, and then by Lemma~\ref{lemma:intersection}, we get that the size of the universe is an upper bound on $\vsrc$ of the graph.

	Every chordal graph is the intersection graph of subtrees of a tree~\cite{GAVRIL197447}. It is also known that the number of vertices of this tree only needs to be at most $|V(G)|-\omega(G)+1$.
	(For completeness, we provide a proof of this in Lemma~\ref{lem:chordal} below.)

For a circular arc graph $G$, consider any set of arcs whose intersection graph is $G$. We now construct a different intersection representation. Take the set of second (considering a clockwise ordering of points) endpoints of all arcs as the universe $\U$. Take $S_i\subseteq \U$ as the set of clockwise endpoints contained in the $i$-th arc. It is easy to see that $G$ is the intersection graph of $\F=\left\{ S_1,S_2,\dots, S_n \right\}$.

Finally, consider $L(G)$. We construct an intersection representation with universe $V(G)$. For each $uv \in E(G)$, let $S_{uv}=\{u,v\}$. Then $L(G)$ is the intersection graph of $\F=\left\{ S_e:e\in E(G) \right\}$.

The (almost) tightness follows from $\vsrc(G) = |V(G)|-1$ and $\vsrc(L(G)) = |V(G)|-2$ for any path $G$. Paths are both chordal and circular-arc.
\qed\end{proof}
\begin{lemma}\label{lem:chordal}
Let $G$ be any chordal graph. Then $\isn(G) \leq |V(G)| - \omega(G) + 1$.
\end{lemma}
\begin{proof}
Let $n = |V(G)|$. Since $G$ is a chordal graph, it has a perfect elimination order~\cite{golumbic2004algorithmic}. In fact, the lexicographic search algorithm that construct a perfect elimination order implies the existence of such an order $v_1, v_2, \dots, v_n$ such that $v_n, v_{n-1}\dots, v_{n-\omega(G)+1}$ forms a maximum clique. For any $i = 1,\ldots,n$, let $G_i=G[{v_n,v_{n-1},\dots, v_i}]$.

We claim that for $n-\omega(G)+1\ge i\ge 1$, $G_i$ can be represented as the intersection graph of subtrees $S_n,S_{n-1},\dots, S_i$ of a tree $T_i$ on at most $n-i-\omega(G)+2$ vertices, and there exists a bijection $f_i$ from vertices of $T_i$ to maximal cliques of $G_i$ such that for each $n\ge k\ge i$, $S_k=T_i[{u\in V(T_i): v_k\in f_i(u)}]$. 

We prove the claim by downwards induction on $i$. The base case is when $i=n-\omega(G)+1$. In this case, the statement follows by taking $T_i$ as a single vertex $u$ and $f_i(u)$ as the clique $\{v_n,v_{n-1},\dots, v_i\}$.

Now, assuming the statement is true for $i$, we prove it for $i-1$. Since $v_{i-1}$ is simplicial in $G_{i-1}$, $N_{G_{i-1}}(v_{i-1})$ is a subset of some maximal clique $C$ of $G_j$. Let $t$ be the vertex in $T_j$ such that $f_j(t)=C$. If all vertices in $C$ are adjacent to $v_{i-1}$ in $G_{i-1}$, then we take $T_{i-1}=T_{i}$, $f_{i-1}(t)=f_i(t)\cup \{v_{i-1}\}$, and $f_{i-1}(t')=f_{i}(t')$ for all $t'\in V(T_{j+1})$ such that $t'\neq t$. It is easy to see that the statement follows in this case.	Now suppose that not all vertices in $C$ are adjacent to $v_{i-1}$ in $G_{i-1}$. Then we take $V(T_{i-1})=V(T_j)\cup \{u\}$ and $E(T_{i-1})=E(T_i)\cup \{u,t\}$ where $u$ is a new vertex introduced with $f_{i-1}(u)=N_{G_{i-1}}[v_{i-1}]$.	For all $u'\in V(T_{i-1})$ such that $u'\neq u$, we take $f_{i-1}(u')=f_{i}(u) $. The statement follows from this construction.

By taking $i=1$ in the statement of the claim, it follows that $G$ can be represented as the intersection graph of subtrees of a tree with at most $n-\omega(G)+1$ vertices.
\qed\end{proof}

In the remainder, we consider a natural generalization of line graphs. A graph is \emph{$k$-perfectly groupable} if the neighborhood of each vertex can be partitioned into $k$ or fewer cliques. It is well known that line graphs are $2$-perfectly groupable. A graph is \emph{$k$-perfectly orientable} if there exists an orientation of its edges such that the outgoing neighbors of each vertex can be partitioned into $k$ or fewer cliques. Clearly, any $k$-perfectly groupable graph is also $k$-perfectly orientable. Many geometric intersection graphs, such as disk graphs, are known to be $k$-perfectly orientable for small~$k$~\cite{kammer}.

\begin{corollary} \label{cor:orientable}
Let $G$ be any $k$-perfectly orientable graph. Then, $\src(G) \leq \vsrc(G) \leq k |V(G)|$.
\end{corollary}
\begin{proof}
Consider any orientation of the edges of $G$ such that the outgoing neighbors of each vertex can be partitioned into $k$ or fewer cliques. For a given vertex $v$, let $C(v)$ denote the set of cliques induced by its outgoing neighbors, where $v$ is added to each of those cliques. Observe that $\bigcup_{v\in V(G)}C(v)$ is an edge clique cover of $G$, because every edge is outgoing from some vertex $v$ and will thus be covered by a clique in $C(v)$. Hence, $\vsrc(G) \leq \ecc(G) \leq k |V(G)|$.
\qed\end{proof}

Since any $k$-perfectly groupable graph is also $k$-perfectly orientable, the above bound also applies to $k$-perfectly groupable graphs. In this context, we prove an interesting converse of the above bound.

\begin{lemma} \label{lem:groupable}
Let $G$ be any graph. If $\vsrc(G) \leq k$, then $G$ is $k$-perfectly groupable.
\end{lemma}
\begin{proof}
Consider an optimal very strong rainbow coloring $\mu$ of $G$. Consider an arbitrary vertex $v$ of $G$ and let $c$ be any color used in $\mu$. Define the set $Q(c)=\left\{ u\in N(v): \mu(vu)=c \right\}$. Suppose there exist two non-adjacent vertices $u,w$ in $Q(c)$. Then $uvw$ is a shortest path between $u$ and $w$, and thus $uv$ and $vw$ cannot have the same color, a contradiction. Hence, for each color $c$ used in $\mu$, $Q(c)$ is a clique. Since the number of colors is at most $k$, the edges incident on $v$ can be covered with at most $k$ cliques. Hence, $G$ is $k$-perfectly groupable.
\qed\end{proof}

\section{Hardness Results}\label{sec:hardness}
The hardness results lean heavily on the combinatorial bounds of the previous section. 
In this section, we use $\hat{G}$ (see the preliminaries for the definition) extensively.
We need the following bound, which strengthens Lemma~\ref{vsrc-l-cp2}.

\begin{lemma}\label{cp-3-vsrc}
Let $G$ be any graph. If $\cp(G) \leq 3$, then $\vsrc(\hat{G})\le 3$.
\end{lemma}
\begin{proof}
	Let $C_1$, $C_2$, and $C_3$ be three cliques into which $V(G)$ is partitioned. We will color $\hat{G}$ with three colors, say $c_1$, $c_2$, and $c_3$, as follows. For each edge with both endpoints in $C_i$ for $1\le i\le 3$, color it with $c_i$. For each edge $vw$ with $v \in C_i$, $w \in C_j$ such that $1 \le i < j \le 3$, color it with $c_k$, where $k \in \{1,2,3\} \setminus \{i,j\}$. Finally, for each edge $\hat{u}v$ with $v\in C_i$ for $1\le i\le 3$, color it with $c_i$. 

Suppose this is not a very strong rainbow coloring of $\hat{G}$. Since the diameter of $\hat{G}$ is at most~$2$, there exists a shortest path $xyz$ with $xy$ and $yz$ having the same color. However, if $xy$ and $yz$ have the same color, at least two of $x,y$ and $z$ are in the same $C_i$ for $1\le i\le 3$ and the third one is either $\hat{u}$ or in $C_i$ itself. Then, we can shortcut $xyz$ by $xz$, a contradiction. 
	 Hence $\vsrc(\hat{G}) \leq 3$.
\qed\end{proof}

\begin{proof}[of Theorem~\ref{thm:hardness}]
	We first prove that {\tvsrc} is NP-complete. We reduce from the NP-hard \textsc{$3$-Coloring} problem~\cite{gary1979computers}. Let $G$ be an instance of \textsc{$3$-Coloring}. Let $H$ be the complement of $G$. We claim that $\vsrc(\hat{H}) = 3$ if and only if $G$ is $3$-colorable. Indeed, if $\vsrc(\hat{H}) \leq 3$, then $\hat{H}$ is $3$-perfectly groupable by Lemma~\ref{lem:groupable}. In particular, the neighborhood of $\hat{u}$ (the universal vertex in $\hat{H}$) can be partitioned into at most $3$ cliques. These cliques induce disjoint independent sets in $G$ that cover $V(G)$, and thus $G$ is $3$-colorable. For the other direction, note that if $G$ is $3$-colorable, then $\cp(H) \leq 3$, and by Lemma~\ref{cp-3-vsrc}, $\vsrc(\hat{H}) \leq 3$.

To prove the hardness of approximation, we recall that there exists a polynomial-time algorithm that takes a {\sat} formula $\psi$ as input and produces a graph $G$ as output such that if $\psi$ is not satisfiable, then $\cp(G) \geq |V(G)|^{1-\varepsilon}$, and if $\psi$ is satisfiable, then $\cp(G)  \leq |V(G)|^{\varepsilon}$~\cite[Proof of Theorem 2]{zuckerman2006linear}. Consider the graph $\hat{G}$ and let $n$ denote the number of its vertices. Then
\begin{align*}\psi\ \mbox{not satisfiable} \Rightarrow \cp(G)\ge (n-1)^{1-\varepsilon} \Rightarrow \vsrc(\hat{G})\ge (n-1)^{1-\varepsilon}\\
\psi\ \mbox{satisfiable} \Rightarrow \cp(G)\le (n-1)^{\varepsilon} \Rightarrow \vsrc(\hat{G})\le (n-1)^{2\varepsilon}\end{align*}
because Lemma~\ref{lem:groupable} implies that $\vsrc(\hat G)\ge \cp(G)$, and 
by Lemma~\ref{vsrc-l-cp2}. The result follows by rescaling $\varepsilon$.
\qed\end{proof}

\section{Algorithm for Cactus Graphs}
\label{sec:polynomial_time_algorithm_for_very_strong_rainbow_coloring_cactus_graphs}
Let $G$ be the input cactus graph. 
We first prove several structural properties of cactus graphs, before presenting the actual algorithm.

\subsection{Definitions and Structural Properties of Cactus Graphs}
We make several structural observations related to cycles.
For a vertex $v$ and a cycle $C$ containing $v$, we define $\subg\left(v,C\right)$ as the vertices of $G$ that are reachable from $v$ without using any edge of~$C$.

\begin{observation}
	\label{sv-partition}
	For any cycle $C$ in $G$, $\left\{ 
	\subg(v,C):v\in V(C) \right\}$ is a partition of $V(G)$.
\end{observation}
From Observation~\ref{sv-partition}, we have that for any fixed $u\in V(G)$ and any fixed cycle $C$ of $G$, there exists a unique vertex $v\in V(C)$ such that $u\in \subg(v,C)$. We denote that unique vertex $v$ by $\g(u,C)$.

\begin{observation}
	\label{path-entry}
	Let $u\in V(G)$ and let $C$ be a cycle in $G$. 
	Let $w\in V(C)$ and let $x_1x_2\dots x_r$ be a path from $u$ to $w$ where $x_1=u$ and $x_r=w$.
	Let $i^*$ be the smallest $i$ such that $x_i\in V(C)$.
	Then, $x_{i^*}=\g(u,C)$.
	In simpler words, any path from $u$ to any vertex in $C$ enters $C$ through $\g(u,C)$.
\end{observation}

\begin{observation}
	\label{edge-entry}
	For any cycle $C$ in $G$ and for any $uv\in E(G)\setminus E(C)$, $\g(u,C)=\g(v,C)$.
\end{observation}
We now consider even cycles. For an edge $uv$ in an even cycle $C$, we define its \emph{opposite edge}, denoted by $\eeopp(uv)$, as the unique edge $xy\in E(C)$ such that $d(u,x)=d(v,y)$. Note that $\eeopp(\eeopp(e))=e$. Call the pair of edges $e$ and $\eeopp(e)$ an \emph{opposite pair}. Each even cycle $C$ has exactly $\frac{|C|}{2}$ opposite pairs. 

\begin{lemma}
	\label{even-edge}
	Let $C$ be an even cycle. 
	For any vertex $x\in V(G)$ and edge $uv\in E(C)$, either there is a shortest path between $x$ and $u$ that contains $uv$ or there is a shortest path between $x$ and $v$ that contains $uv$.
\end{lemma}
\begin{proof}
	Let $w=\g(x,C)$.
	Then, $w$ cannot be equidistant from $u$ and $v$, because otherwise $C$ is an odd cycle.
	Suppose that $d(w,u)<d(w,v)$.
	Then a shortest path from $w$ to $u$ appended with the edge $uv$ gives a shortest path between $w$ and $v$.
	Now, due to Observation~\ref{path-entry}, if we append a shortest path between $x$ and $w$ with a shortest path between $w$ and $v$, we get a shortest path between $x$ and $v$.
	Thus there is a shortest path between $x$ and $v$ that contains $uv$. If $d(w,u)>d(w,v)$, then we get the other conclusion of the lemma.
\qed\end{proof}
We then consider odd cycles in more detail. For any edge $e$ in an odd cycle $C$, there is a unique vertex in $C$, which is equidistant from both endpoints of~$e$. We call this vertex the \emph{opposite vertex} of $e$ and denote it as $\vopp(e)$. We call $\os(e) = G[\subg(\vopp(e),C)]$ the \emph{opposite subgraph} of $e$. See Figure~\ref{fig:cactus}.

\begin{lemma}
	\label{odd-edge}
	Let $C$ be an odd cycle and $uv\in E(C)$.
	For any vertex $x\in V(G)\setminus  V(\os(uv)) $, either there is a shortest path between $x$ and $u$ that contains $uv$ or there is a shortest path between $x$ and $v$ that contains $uv$.
\end{lemma}
\begin{proof}
	Let $w=\g(x,C)$.
	Since $x\notin V(\os(uv))$, $w\neq \vopp(uv)$.
	Hence, $w$ cannot be equidistant from $u$ and $v$.
	So, the same arguments as in Lemma~\ref{even-edge} complete the proof.
\qed\end{proof}


\begin{lemma} \label{opp-exists}
Let $e$ be any edge in an odd cycle of $G$ for which $\vopp(e)$ has degree more than~$2$. Then $\os(e)$ contains a bridge, or an even cycle, or an edge $e'$ in an odd cycle for which $\vopp(e')$ has degree~$2$.
\end{lemma}
\begin{proof}
Suppose this is not the case. We define a sequence $e_1,e_2,\dots$ of edges by the following procedure. Let $e_1=e$. Given $e_i$, we define $e_{i+1}$ as follows. By assumption and the definition of cactus graphs, $e_i$ is contained in an odd cycle, which we denote by $C_i$, and $\vopp(e_i)$ has degree more than~$2$. Choose $e_{i+1}$ as any edge incident on $\vopp(e_i)$ that is not in $C_i$. However, observe that $\os(e_{i+1}) \subset \os(e_i)$ by the choice of $e_{i+1}$. Hence, this is an infinite sequence, which contradicts the finiteness of $E(G)$.
\qed\end{proof}


\subsection{Properties of Very Strong Rainbow Colorings of Cactus Graphs}
We initially partition the edges of $G$ into three sets: $\et,$ $\eeven$, and $\eodd$.
The set $\et$ consists of those edges that are not in any cycle.
In other words, $\et$ is the set of bridges in $G$.
By definition, each of the remaining edges is part of exactly one cycle. 
We define $\eeven$ as the set of all edges that belong to an even cycle, and
$\eodd$ as the set of all edges that belong to an odd cycle. Note that $\et,$ $\eeven$, and $\eodd$ indeed induce a partition of $E(G)$.
We then partition $\eodd$ into two sets: $\eopp$ and $\erem$. An edge $e\in \eodd$ is in $\eopp$ if $\vopp(e)$ is not a degree-2 vertex and in $\erem$ otherwise. See Figure~\ref{fig:cactus}. We analyze each of these sets in turn, and argue how an optimal {\cvsrc} might color them. 

\begin{figure}[tbp]
\centering
	\includegraphics[scale=0.5]{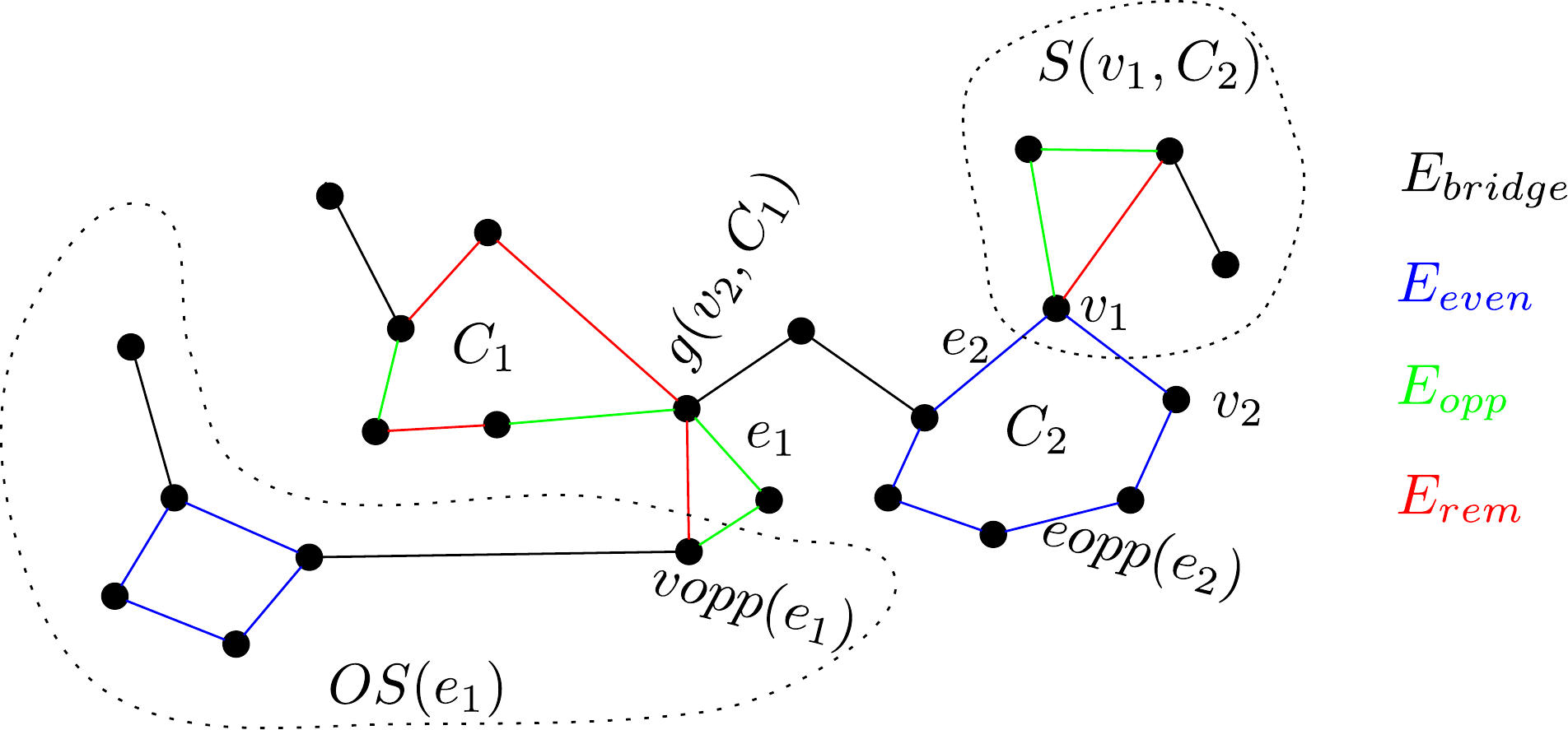}
	\caption{An example of a cactus graph and related definitions.} 
	\label{fig:cactus}	
\end{figure}

Two edges $e_1$ and $e_2$ are called \emph{conflicting} if there is a shortest path in the graph which contains both $e_1$ and $e_2$. Two conflicting edges must have different colors in any \cvsrc. We now exhibit several classes of conflicting pairs of edges.

\begin{lemma} \label{bridge-bridge}\footnote{This lemma holds for any graph, not necessarily cactus}
Any {\cvsrc} of $G$ colors the edges of $\et$ with distinct colors.
\end{lemma}
\begin{proof}
Consider $uv,xy\in \et$.	
We prove that $uv$ and $xy$ are conflicting, i.e. there is a shortest path in $G$  which contains both $uv$ and $xy$. 
Since $uv$ is a bridge, we can assume without loss of generality that any path between $u$ and $y$ uses the edge $uv$.
Similarly, since $xy$ is a bridge, we can assume without loss of generality that any path between $y$ and $u$ uses the edge $xy$.
Hence, the shortest path from $u$ to $y$ uses both $uv$ and $xy$. 
Hence, $uv$ and $xy$ are conflicting.
\qed\end{proof}

\begin{lemma}
	\label{bridge-even}
	Let $e_1\in \et$ and $e_2\in \eeven$. Then any {\cvsrc} of $G$ colors $e_1$ and $e_2$ with different colors.
\end{lemma}
\begin{proof}
Let $C$ be the cycle containing $e_2$. Let $e_1=xy$ and $e_2=uv$.
	Since $xy$ is a bridge, we can assume w.l.o.g.\ that any path from $x$ to any vertex in $C$ contains $xy$.
	Due to Lemma~\ref{even-edge}, we can assume w.l.o.g.\ that there is a shortest path from $x$ to $v$ that contains $uv$.
	Thus we have a shortest path which contains both $uv$ and $xy$, which means that $uv$ and $xy$ are conflicting.
\qed\end{proof}


\begin{observation} \label{base-even}
Let $e_1$ and $e_2$ be edges in an even cycle $C$ of $G$ such that $e_1 \not= \eeopp(e_2)$. Then any {\cvsrc} of $G$ colors $e_1$ and $e_2$ with different colors.
\end{observation}

\begin{lemma}
	\label{even-even}
	Let $e_1$ and $e_2$ be edges in two different even cycles $C_1$ and $C_2$ of $G$. 
	Then any {\cvsrc} of $G$ colors $uv$ and $xy$ with different colors.
\end{lemma}
\begin{proof}
Let $e_1 = uv$ and $e_2 = xy$.
	Let $z=\g(u,C_2)$ and $w=\g(x,C_1)$.
	By Observation~\ref{edge-entry}, $\g(v,C_2)=z$ and $\g(y,C_1)=w$.
	Due to Lemma~\ref{even-edge}, we can assume w.l.o.g.\ that there is a shortest path $P_1$ between $z$ and $x$ containing $xy$ and that there is a shortest path $P_2$ between $w$ and $u$ containing $uv$.
	Let $P_3$ be a shortest path between $w$ and $z$.
	Then $P_1\cup P_3 \cup P_2$ gives a shortest path between $u$ and $x$ that contains both $uv$ and $xy$.
	Hence, $e_1$ and $e_2$ are conflicting.
\qed\end{proof}

\begin{lemma} \label{erem-no-repeat}
Let $e_1\in \et\cup \eeven$ and $e_2\in \erem$. Then any {\cvsrc} of $G$ colors $e_1$ and $e_2$ with different colors.
\end{lemma}
\begin{proof}
Let $e_1=xy$ and $e_2=uv$, let $C$ be the odd cycle containing $e_2$, and let $w=\g(x,C)$. By Observation~\ref{edge-entry}, $w=\g(y,C)$. In other words, $x,y\in \subg(w,C)$. Note that $w$ is not a degree-$2$ vertex, because there are at least two vertices in $\subg(w,C)$.	Hence, $w\neq \vopp(uv)$ by the definition of $\erem$. Hence, by Lemma~\ref{odd-edge}, w.l.o.g.\ there is a shortest path $P_1$ from $w$ to $u$ that contains $uv$.

We now consider two cases, depending on whether $e_1 \in \et$ or $e_1 \in \eeven$. First, suppose that $e_1\in \et$. Since $xy$ is a bridge, we can assume w.l.o.g.\ that any shortest path from $x$ to $w$ contains $xy$. Let $P_2$ be such a shortest path. By Observation~\ref{path-entry}, if we append a shortest path from $x$ to $w$ with a shortest path from $w$ to $u$, we get a shortest path from $x$ to $u$. Thus, $P_1\cup P_2$ is a shortest path from $x$ to $u$ containing $xy$ and $uv$. Hence, $e_1$ and $e_2$ are conflicting.

Suppose that $e_1\in \eeven$. Let $C'$ be the even cycle containing $e_1$.
Let $z=\g(v,C')$.
	From Lemma~\ref{even-edge}, we can assume w.l.o.g.\ that there is a shortest path from $z$ to $x$ that contains $xy$.
	Let this shortest path be $P_3$.
Let $P_4$ be a shortest path between $w$ and $z$.
By Observation~\ref{path-entry}, $P_3\cup P_4\cup P_1$ is a shortest path between $x$ and $u$ that contains $xy$ and $uv$. Hence, $e_1$ and $e_2$ are conflicting.
\qed\end{proof}

\begin{lemma} 
	\label{erem-oddcycle}
Let $C_1$ and $C_2$ be two distinct odd cycles and let $e_1 \in E(C_1)\cap \erem$ and $e_2 \in E(C_2)\cap \erem$. Then any {\cvsrc} of $G$ colors $e_1$ and $e_2$ with different colors.
\end{lemma}
\begin{proof}
Let $e_1=xy$ and $e_2=uv$, and let $w=\g(x,C_2)$. By Observation~\ref{edge-entry}, $w=\g(y,C_2)$. Let $z=\g(u,C_1)$. By Observation~\ref{edge-entry}, $z=\g(v,C_1)$. That is, $x,y\in \subg(w,C_2)$ and $u,v\in \subg(z,C_1)$. Note that $w$ and $z$ are not degree-$2$ vertices, because there are at least two vertices in $\subg(w,C_2)$ and $\subg(z,C_1)$. Hence, $w\neq \vopp(uv)$ and $z\neq \vopp(xy)$ by the definition of $\erem$.	Hence, by Lemma~\ref{odd-edge}, we can assume w.l.o.g.\ that there is a shortest path $P_1$ from $u$ to $w$ that contains $uv$ and there is a shortest path $P_2$ from $z$ to $x$ that contains $xy$.
	Let $P_3$ be a shortest path from $w$ to $z$.
By Observation~\ref{path-entry}, 
$P_1\cup P_2\cup P_3$ is a shortest path from $x$ to $u$ containing $xy$ and $uv$. Hence, $e_1$ and $e_2$ are conflicting.
\qed\end{proof}

Finally, we prove the existence of some non-conflicting pairs of edges.

\begin{lemma} \label{opp-valid}
For any $e_1\in \eopp$ and $e_2\in \os\left( e_1 \right)$, $e_1$ and $e_2$ are not conflicting.
\end{lemma}
\begin{proof}
Let $e_1=uv$, $e_2=xy$, and let $C$ be the odd cycle containing $e_1$. For sake of contradiction, suppose that $uv$ and $xy$ are conflicting.	Assume w.l.o.g.\ that there is a shortest path $P$ from $x$ to $v$ which contains $uv$ and $xy$. From Observation~\ref{path-entry}, $P$ contains a subpath $P'$ from $\g(x,C)$ to $v$. Clearly, $P'$ contains $uv$. Also, $\g(x,C)=\vopp(uv)$, because $x\in \os(uv)$. However, recall that $\vopp(uv)$ is equidistant from $u$ and $v$. Hence, any shortest path from $\vopp(uv)$ to $v$ does not contain $uv$, which contradicts the existence of $P'$.
\qed\end{proof}

\subsection{Algorithm}
Based on the results of the previous two subsections, we now describe the algorithm for cactus graphs.
First, we color the edges of $\et$ with unique colors. By Lemma~\ref{bridge-bridge}, no VSRC can use less colors to color $\et$.

Next, we color the edges in $\eeven$ using colors that are distinct from those we used before. This will not harm the optimality of the constructed coloring, because of Lemma~\ref{bridge-even}. Moreover, we use different colors for different even cycles, which does not harm optimality by Lemma~\ref{even-even}. We then introduce a set of $\frac{|C|}{2}$ new colors for each even cycle $C$. For an opposite pair, we use the same color, and we color each opposite pair with a different color. Thus we use $\frac{|C|}{2}$ colors for each even cycle $C$. By Observation~\ref{base-even}, no VSRC can use less colors to color $C$.

Next, we will color the edges in $\erem$ using colors that are distinct from those we used before. This will not harm the optimality of the constructed coloring, because of Lemma~\ref{erem-no-repeat}. For each odd cycle, we use a different set of colors. This will not harm the optimality of the constructed coloring, because of Lemma~\ref{erem-oddcycle}.

For each odd cycle $C$, we construct an auxiliary graph $H_C$ for $\erem\cap C$ as follows. Let $V(H_C)=\erem\cap C$ and let $E(H_C)=\{ e_1e_2 : e_1,e_2\in V(H_C)$; $e_1$ and $e_2$ are not conflicting in $G \}$.

\begin{lemma}\label{conflict-bip}
$\Delta(H_C)\le 2$.
\end{lemma}
\begin{proof}
It is easy to observe that in any odd cycle $C$, for any $e\in E(C)$, there are only two other edges in $C$ that are not conflicting with $e$.
\qed\end{proof}

Let $M_C$ be a maximum matching of $H_C$. We can compute $M_C$ in linear time, since $\Delta(H_C) \leq 2$. For an $e_1e_2\in M_C$, color $e_1$ and $e_2$ with the same, new color.
Then color each $e\in \erem\cap C$ that is unmatched in $M_C$, each using a new color.

\begin{lemma}
	\label{erem-procedure}
The procedure for coloring $\erem \cap C$ gives a coloring of the edges in $\erem\cap C$ such that no conflicting edges are colored the same. Moreover, no {\cvsrc} of $G$ can use less colors to color $\erem\cap C$ than used by the above procedure.
\end{lemma}
\begin{proof}
Suppose two conflicting edges $e_1, e_2 \in \erem \cap C$ were colored the same. Then the corresponding vertices $e_1$ and $e_2$ were matched to each other in $M_C$. Hence, $e_1$ and $e_2$ are adjacent in $H_C$, meaning that $e_1$ and $e_2$ did not conflict each other in $G$, which is a contradiction. Hence, we have proved that no conflicting edges were given the same color by the procedure. 

Now, consider any VSRC $\mu$ of $G$ which colored $\erem\cap C$ with fewer colors than by our procedure. Observe that for any edge $e$ in an odd cycle, there are only two other edges (say $e_a$ and $e_b$) that are not conflicting with $e$. Moreover, $e_a$ and $e_b$ are conflicting with each other. This means that $\mu$ can use each color for at most two edges of $\erem\cap C$. Suppose there are $k_1$ colors that are assigned to two edges in $\erem\cap C$ by $\mu$. Each pair of edges colored the same should be non-conflicting and hence have an edge between them in $H_C$. So, taking all pairs colored the same induces a matching of size $k_1$ of $H_C$. Then $k_1\le |M_C|$, because $M_C$ is a maximum matching of $H_C$. But then the number of colors used by $\mu$ is equal to $k_1+(|\erem\cap C|-2k_1)=|\erem\cap C|-k_1$. The number of colors used by our procedure is $|M_C|+|\erem\cap C|-2|M_C|=|\erem\cap C|-|M_C|\le |\erem\cap C|-k_1$. Hence, we use at most the number of colors used by $\mu$.
\qed\end{proof}


Finally, we color the edges of $\eopp$ without introducing new colors. Indeed, for every $e \in \eopp$, it follows from Lemma~\ref{opp-exists} that there exists an edge $e'\in E(\os(e))\cap \left( \et\cup\eeven\cup \erem \right)$, which does not conflict with $e$ by Lemma~\ref{opp-valid}. Since $e'$ is already colored, say by color $c$, then we can simply re-use that color $c$ for $e$. Indeed, suppose for sake of contradiction that there is a shortest path $P$ between two vertices $x,y$ that contains $e$ and that contains another edge $e''$ using the color $c$. By Lemma~\ref{opp-valid}, $e'' \not\in \os(e)$. This implies that $e'' \not\in \et \cup \eeven \cup \erem$ by the choice of $c$ and the construction of the coloring. Hence, $e'' \in \eopp$. However, by a similar argument, $e''$ can only receive color $c$ if $e' \in \os(e'')$. But then $\os(e) \subseteq \os(e'')$ or $\os(e'') \subseteq \os(e)$, and thus $e$ and $e''$ are not conflicting by Lemma~\ref{opp-valid}, a contradiction to the existence of $P$.

\begin{proof}[of Theorem~\ref{thm:cactus}]
It follows from the above discussion that the constructed coloring is a very strong rainbow coloring of $G$. Moreover, it uses $\vsrc(G)$ colors. Clearly, the coloring can be computed in polynomial time.
\qed\end{proof}

\section{Other Algorithmic Results}
\label{sec:other_results}
In this section, we first show that {\twovsrc} can be solved in polynomial time. Then we show that {\kvsrc} is fixed parameter tractable when parameterized by $k+\tw(G)$, where $\tw(G)$ denotes the treewidth of $G$. 

For proving both the results, we use an auxillary graph $G'$ defined as follows: add a vertex $v_e$ to $G'$ for each edge $e$ in $G$; add an edge between vertices $v_{e_1}$ and $v_{e_2}$ in $G'$ if and only if edges $e_1$ and $e_2$ appear together in some shortest path of $G$. The latter condition can be easily checked in polynomial time. 
Observe that $\vsrc(G) \leq k$ if and only if $G'$ admits a proper $k$-coloring.
Since {\twocolor} is solvable in polynomial time, this implies that {\twovsrc} is polynomial time solvable and hence we have proved Proposition~\ref{prp:two}.

It is worth noting that the chromatic number of the auxiliary graph $G'$ constructed in the above proof always corresponds to the very strong rainbow connection number of $G$. However, in the transformation from $G$ to $G'$, we lose a significant amount of structural information. For example, if $G$ is a path or a star ($\tw(G)=1$), then $G'$ is a clique ($\tw(G') = |V(G'| -1 = |V(G)|-2$), where we use $\tw(G)$ to denote the treewidth of $G$. However, if $\vsrc(G) \leq k$, then we can prove that $|V(G')| \leq k^{(k+1)} \cdot (\tw(G)+1)^{(k+1)}$ as shown below.

\begin{lemma}\label{lem:twbound}
Let $G$ be any connected graph and let $\vsrc(G) \leq k$ and $\tw(G) \leq t-1$. Then $\Delta(G) \leq kt$ and $|V(G)| \leq (kt)^k$.
\end{lemma}
\begin{proof}
By Lemma~\ref{lem:groupable}, the fact $\vsrc(G) \leq k$ implies that $G$ is $k$-perfectly groupable. Hence, the neighborhood of each vertex can be partitioned into $k$ or fewer cliques. Since $\tw(G) \leq t-1$, each clique of $G$ has size at most $t$~\cite{robsey-tw}. Hence, $\Delta(G) \leq kt$. Now observe that $\vsrc(G) \leq k$ implies that the diameter of $G$ is at most $k$. Combined, these two facts imply that $|V(G)| \leq (kt)^k$.
\qed\end{proof}

\begin{proof}[of Theorem~\ref{fpt-kt}]
Again, let $\vsrc(G) \leq k$ and $\tw(G) \leq t-1$.
We now construct the auxiliary graph $G'$ as above.  Now, we only need to compute the chromatic number of $G'$.
We aim to use the algorithm by Bj\"orklund~\etal~\cite{bjorklund} which computes the chromatic number of a graph on $n$ vertices in $2^n n^{\order(1)}$ time. To bound $|V(G')|$, we observe that by Lemma~\ref{lem:twbound}, $|V(G)| \leq (kt)^k$ and $\Delta(G) \leq kt$. Hence, $|V(G')| = |E(G)| \leq (kt)^{(k+1)}$. Therefore, the chromatic number of $G'$, and thereby $\vsrc(G)$, can be determined in $\order(2^{(kt)^{(k+1)}} (kt)^{\order(k+1)})$ time.
\qed\end{proof}
\bibliography{vsrc-short}

\end{document}